\documentclass[conference,10pt]{IEEEtran}
\usepackage[T1]{fontenc} 
\usepackage[english]{babel} 
\usepackage{amsmath,amsfonts,amsthm,amssymb} 
\usepackage{algorithm}
\usepackage[noend]{algpseudocode}
\usepackage{enumitem}
\usepackage{cite,graphicx,float}
\usepackage{subcaption}
\usepackage{caption}
\usepackage[figurewithin=none]{caption}
\usepackage{mathtools, nccmath}
\usepackage{color}

\usepackage{fancyhdr} 
\newtheorem{theorem}{Theorem}
\newtheorem{lemma}{Lemma}
\newtheorem{remark}{Remark}
\def\bm#1{\mathbf{#1}}

\title{	\huge Optimizing Throughput in a MIMO System with a Self-sustained Relay and Non-uniform Power Splitting\\ 
}

\author{Rafia Malik and Mai Vu\\
Department of Electrical and Computer Engineering, Tufts University, Medford, MA, USA\\
Email: rafia.malik@tufts.edu, mai.vu@tufts.edu} 

\date{\normalsize{March 17, 2017}} 
\begin{document}

\maketitle 

\begin{abstract}
We present a novel approach to maximizing the transmission rate in a MIMO relay system, where all nodes are equipped with multiple antennas and the relay is self-sustained by harvesting energy. We formulate an optimization problem and use dual-characterization to derive a closed-form solution for the optimal power splitting ratio and precoding design. We propose an efficient primal-dual algorithm to jointly optimize the power allocation at source and relay for transmission and the power splitting at relay for energy harvesting, and show that using non-uniform power splitting is optimal. Numerical results demonstrate the significant rate gain of non-uniform power splitting over traditional uniform splitting especially at low source transmit power. We also analyze our algorithm numerically and demonstrate its efficiency at reducing the run-time by several orders of magnitudes compared to a standard solver, \textcolor{blue}{and existing algorithms in literature}.
\end{abstract}

\section{Introduction}

With the rapid evolution of wireless networks in recent years, energy efficiency is now an important figure of merit for the design of next generation communication systems~\cite{Krikidis2014}. For communication networks, we can envision future networks to employ relays capable of providing cooperation in terms of both information and energy.  Prior works in joint information and energy transfer have considered systems with multiple antennas at the transmitting nodes and single antennas at the receiving nodes~\cite{Rui2015}\cite{Qiao2018}, or investigated systems with multiple antennas only at the relay node~\cite{Mohammadi2015}. While such systems can provide throughput gains over single-antenna ones, a MIMO system with multiple antennas at all nodes can further enhance the system's performance not only in terms of the achievable rate, but also the harvested energy. \textcolor{blue}{New research efforts also explore energy harvesting and communication using massive MIMO systems~\cite{Xu2018}.}

Two practical designs for MIMO channels with energy harvesting and information decoding exist in literature - Time Switching (TS) and Power Splitting (PS)~\cite{Alsharoa2017}~\cite{Rui2013}. In general, power splitting has reduced transmission delay and increased spectral efficiency compared to time switching. To simplify the system model, a common assumption in the literature for power splitting is that it is uniform among all antennas at the relay node~\cite{Alouini2016}~\cite{Khandaker2016}. While uniform power splitting can achieve the same rate as non-uniform splitting for systems with multiple antennas at a single node~\cite{Rui2013}, for a MIMO system, it is non-optimal. 

In this letter, we present a novel optimization problem to characterize a two hop, decode-and-forward MIMO relay system with simultaneous wireless information and power transfer using non-uniform power splitting. We formulate the problem to jointly optimize both the pre-coding designs at source and relay and the power splitting ratios at relay. Using the Lagrangian dual problem and a sub-gradient method, we solve the resulting convex optimization problem through a proposed primal-dual algorithm. We contrast the performance of our algorithm with existing algorithms, and compare the throughput of our proposed non-uniform power splitting scheme to the traditional uniform splitting. \textcolor{blue}{This comparison is done for an increasing number of antennas in the MIMO system, and the extremely low convergence time of our algorithm corroborates how it may be computationally feasible to apply it to massive MIMO systems.}

\section{System Model}
\subsection{Channel Model}
A half-duplex, decode-and-forward two-hop MIMO relay system is considered, where all three nodes, source, relay and destination, are equipped with multiple antennas. For this letter, we assume that the direct transmission link suffers from significant path loss and fading, such that the relay channel is always used for data transmission from source to destination. The source S, relay R and destination D are equipped with $N_s$, $N_r$ and $N_d$ antennas respectively and the S-R and R-D Rayleigh fading channels are modeled by matrices $\textbf{H} \in \mathbb{C}^{N_r \times N_s}$ and $\textbf{G} \in \mathbb{C}^{N_d \times N_r}$ respectively. We assume a quasi-static fading channel with local channel state information (CSI) available to both the transmit and receive nodes to reveal theoretical bounds of the problem considered.  

\subsection{Relay Model}
The relay is a self-sustained node, employing a harvest-use policy. It is equipped with two rechargeable batteries which store harvested energy and supply power using a TS scheme which caters for the half-duplicity of energy transfer.

The power of the received signal at each antenna of the relay is divided for Energy Harvesting (EH) and Information Decoding (ID), according to the power splitting ratio $\rho_i \in (0,1) \ \forall i \in [1,N_r]$, and we define $\boldsymbol{\Lambda_{\rho}} = \textbf{diag} (\rho_1 ... \rho_{N_r})$. The received signal at the relay is given by $\mathbf{y_{ID}}$ for information decoding and $\mathbf{y_{EH}}$ for energy harvesting as   
\begin{align}
&\mathbf{y_{ID} = (I- \Lambda_{\rho})}^{1/2} \mathbf{Hx_s + z_r}, \ \mathbf{y_{EH} }= \boldsymbol{\Lambda_{\rho}}^{1/2} \mathbf{Hx_s}
\end{align}
where $\mathbf{x_s}$ is the source transmitted signal and $\mathbf{z_r} \sim \mathcal{CN}(\bm 0,\sigma_r^2\bm I)$ denotes the traditional receiver noise. Here we adopt the standard assumption of perfect power splitter~\cite{Rui2013}, and hence include no noise term for $\mathbf{y_{EH}}$. 

Assuming that on average the energy consumed is equal to the energy harvested to prevent energy-outages in the data transmission phase~\cite{Ulukus2015}, the transmission power for the relay, $P_r$, is then equal to the harvested power, $P_h$. The relay transmission power can be written as
\begin{equation}\label{Pr}
P_r = P_h = \frac{E_h}{T/2} = \frac{\eta \text{tr} \Big(\mathbf{\Lambda_{\rho} HW_sH^\ast}\Big)}{T/2}
\end{equation}
where $\eta \triangleq \eta_c \eta_d$, where $\eta_c \in [0,1]$ is the efficiency for energy conversion (from RF to DC) and battery charging, $\eta_d \in [0,1]$ is the utilizing efficiency for battery discharging and $\boldsymbol{W_s}$ is the source covariance matrix. For this letter, similar to~\cite{Rui2013} we assume without loss of generality that $\eta = 1$, unless otherwise stated. We assume a half block (time slot) of unit duration $T_s/2 = 1$, and hence use the expressions for harvested energy $E_h$ and harvested power $P_h$ interchangeably in this letter.  

\textcolor{blue}{The regenerative relay employs a decode-forward multi-hop relaying scheme~\cite{Gamal2011}. The relay recovers the message received from the sender
in each block and re-transmits it in the following block. The signal received at the destination is as given below
\begin{equation}
\mathbf{y_d = Gx_r + z_d}, 
\end{equation}
where $\mathbf{z_d} \sim \mathcal{CN}(\bm 0,\sigma_d^2 \bm I)$ is the additive noise at the destination. The average transmit power constraint on the relay is related to the harvested power in (2) as $\text{tr}\Big(\mathbf{W_r}\Big) \leq P_r$. The receiver then decodes on the signal received from the relay to recover the information transmitted from the source.}

\subsection{Achievable Transmission Rate}
An achievable rate for the multi-hop relay channel is given as\cite[p.~387]{Gamal2011}.
\begin{align*}
&R = \max_{p(x_s) p(x_r)} \min\{I(X_s; Y_{ID} \vert X_r), I (X_r; Y_d)\}\\
&\hphantom{R} = \min \{\max_{p(x_s)} R_{S-R}, \max_{p(x_r)} R_{R-D} \}
\end{align*}
where the second expression follows from application of the first expression to the considered two-hop cascaded S-R and R-D channel and $R_{S-R}$ and $R_{R-D}$ are achievable rates of the first and second hop, respectively. Using the signal model in (1) and (3), assuming optimal Gaussian transmit signals, the achievable rate for S-D transmission, in bps/Hz, is then 
\textcolor{blue}{\begin{align*}\label{csb}
&R\big (\boldsymbol{\Lambda_{\rho}},\mathbf{W_s,W_r} \big ) = {\text{min}}\left \{ \underset{{\mathbf{W_s}} } \max \ R_{S-R}, \underset{{\mathbf{W_r}} } \max \ R_{R-D} \right \}  \tag{4}\\
&\medmath{= \text{min} \left \{ \max_{\mathbf{W_s}\boldsymbol{\Lambda_{\rho}}}\frac{1}{2} \text{log}_2 \left | \mathbf{I} + \frac{\mathbf{(I - \Lambda_{\rho})HW_sH^\ast}}{\sigma_r^2} \right | , \max_{\mathbf{W_r}}\frac{1}{2} \text{log}_2 \left | \mathbf{I} + \frac{\mathbf{GW_rG^\ast}}{\sigma_d^2} \right | \right \}}
\end{align*}}
\section{Throughput Optimization Problem formulation}
The maximization of the achievable end-to-end transmission rate is formulated as an optimization problem given below, where (5b) and (5c) are the transmit power constraints at the source and relay, respectively.
\begin{align*}
&(\text {P}):\underset {\boldsymbol{\Lambda_{\rho}}, \boldsymbol {W}_{s},\boldsymbol {W}_{r}}{\max } ~ R\left ({\boldsymbol{\Lambda_{\rho}}, \boldsymbol {W}_{s},\boldsymbol {W}_{r}}\right ), \tag{5a}\\
&\hphantom {(\text {P}1):}\text {s.t.}~ \mathop {\mathrm {{\text{tr}}}}\nolimits \left ({\boldsymbol {W}_{s}}\right ) \leq P_{s}, \tag{5b}\\
&\hphantom {(\text {P}1):\text {s.t.}~} \mathop {\mathrm {{\text{tr}}}}\nolimits \left ({\boldsymbol {W}_{r}}\right ) \leq P_{r}, \tag{5c}\\
&\hphantom {(\text {P}1):\text {s.t.}~}\boldsymbol {W}_{s} \succeq 0,\quad \boldsymbol {W}_{r} \succeq 0. \tag{5d}
\end{align*}

Here $R\left ({\boldsymbol{\Lambda_{\rho}}, \boldsymbol {W}_{s},\boldsymbol {W}_{r}}\right )$ is the end to end transmission rate, as defined in~(\ref{csb}), where each hop, S-R and R-D, is a point to point MIMO channel. Therefore, the maximum rate for the Gaussian vector channel in each hop is achieved using spatial multiplexing~\cite{Molisch2012}. With CSI at the transmitters, the optimal source covariance matrix has the form $\mathbf{W_s^\star = V_H \Lambda_s V_H^\ast}$, where $\mathbf{V_H}$ is obtained from the Singular Value Decomposition (SVD) of the channel matrix for the information decoding receiver at the relay, $\mathbf{(I - \Lambda_{\rho})H}$, and $\mathbf{\Lambda_s} = \textbf{diag } (p_1 ... p_{K_1})$, with the diagonal elements obtained from water-filling~\cite{Cover2012}. Similarly, for the R-D channel, $\mathbf{W_r^\star = V_G \Lambda_r V_G^\ast}$, with $\mathbf{\Lambda_r} = \textbf{diag }(q_1 ...q_{K_2})$ and $\mathbf{V_G}$ is obtained from the SVD of the R-D channel matrix $\mathbf{G}$. Here $K_1 \text{ and } K_2$ are the number of active channels corresponding to the non-zero singular values of the channel matrices $\bm H \text{ and } \bm G$, respectively.

We now rewrite the optimization problem (P), as the equivalent problem (P-eq) given below. 
\begin{align*}\label{P_new}
&(\text {P-eq}) : ~\underset {R,\boldsymbol{\rho, p, q}}{\max } ~ R \tag{6}\\
&\hphantom {(\text {P}) : ~}\text {s.t.}~ R \leq \frac{1}{2} \sum_{i=1}^{K_1} \log_2(1+(1-\rho_i)p_i\lambda_{H,i}) \tag{6a}\\
&\hphantom {(\text {P}) : ~\text {s.t.}~} R \leq \frac{1}{2} \sum_{i=1}^{K_2} \log_2(1+q_i\lambda_{G,i}) \tag{6b} \\ 
&\hphantom {(\text {P}) : ~\text {s.t.}~} \sum_{i=1}^{K_1} p_i \leq P_s \tag{6c} \\
&\hphantom {(\text {P}) : ~\text {s.t.}~} \sum_{i=1}^{K_2} q_i \leq \sum_{i=1}^{K_1} \rho_i p_i \lambda_{H,i} \tag{6d}
\end{align*}
\textcolor{blue}{Here $\lambda_{H,i}$ and $\lambda_{G,i}$ are the eigenvalues for the S-R and R-D channels respectively.} Implicit constraints not mentioned in (P-eq) are $p_i \geq 0 \ \forall \ i, q_j \geq 0 \ \forall \ j$ and $ \ 0 \leq \rho_k \leq 1 \ \forall \ k$. We impose these constraints as boundary conditions in our algorithm later on. Since the noise power is normalized to unity, we commonly refer to the power constraint $P_s$ in (6c) as the SNR.
\begin{lemma}
The optimization problem (P-eq) is jointly convex in the optimizing variables.
\end{lemma}
\begin{proof}
While it can be readily seen that the objective function $R$ is linear, constraint (6b) is affine in $R$ and convex in $q_i$, constraint (6c) is affine in $p_i$, however, joint convexity for (6a) and (6d) in the optimizing variables $(p_i, q_i, \rho_i)$ needs to be established. For constraints (6a) and (6d), the left hand side is linear in $R$ and $q_i$ respectively, so we consider the right hand side. For (6a) we define $g_1(\rho_i, p_i) = 1 + (1 - \rho_i) \lambda_{H,i} p_i$ which is neither convex nor concave, since its Hessian is indefinite with eigenvalues, $\pm \lambda_{H,i}$. The superlevel sets $\{(\rho_i, p_i) \in \mathbf{R}_{+}^2), g_1(\rho_i p_i) \geq t \}$, are convex $\forall \ t$, which makes $g_1(\rho_i p_i)$ a quasi-concave function~\cite{Boyd2004}. Applying the implicit constraints; $\rho_i, p_i \geq 0$ then $\textbf{dom }g_1(\rho_i, p_i) \equiv \mathbf{R}_{+}^2$. The composition function, $f = h \circ g_1$ in constraint (6a), of the non-decreasing function $h(x) = \log(x)$, and quasi-concave function $g_1(\rho_i, p_i)$, is therefore quasiconcave. Similarly for (6d), we define $g_2(\rho_i, p_i) = \rho_i \lambda_{H,i} p_i$, which is quasi-concave with convex superlevel sets in $\mathbf{R}_{+}^2$. The problem then is a maximization of a convex function, over convex sets and convex sub-level sets, and is therefore a convex optimization problem.
\end{proof}
\subsection{Dual Problem Formulation}
We use the Lagrangian dual method to solve for the optimal solution of the primal problem (P-eq), where the Lagrangian is formulated as given below
\begin{align*}
&\mathcal{L}(R,\boldsymbol{\rho},\bm p,\bm q,\alpha,\beta,\nu,\mu) = R - \alpha(R-R_1) - \beta(R-R_2) \\
&\hphantom{\mathcal{L}(R,\rho,\bm p,\tilde{\bm p} }- \nu\left (\sum_{i=1}^{K_1} p_i - P_s\right )- \mu\left (\sum_{i=1}^{K_2} q_i - \sum_{i=1}^{K_1} \rho_i p_i \lambda_{H,i}\right )
\end{align*}
\textcolor{blue}{Here $\alpha, \beta, \nu, \mu$ are the dual variables associated with constraints (6a)-(6d) respectively.} Since $\mathcal{L}$ is a linear function of $R$ we use the optimality condition and set $\nabla_R \mathcal{L} = 0$, which gives us $1-\alpha-\beta = 0 \implies \beta = 1 - \alpha$. This is also intuitive, since the transmission rate, $R$, is equal to the minimum of the rate of the two hops, and is equal to either $R_1$ or $R_2$ or both when $R_1 = R_2$. Substituting $\beta = 1-\alpha$, we can rewrite the Lagrangian as,
\begin{align*}\label{Lagrangian}
&\mathcal{L}(\boldsymbol{\rho},\bm p,\bm q,\alpha,\nu,\mu) = \alpha R_1 + (1-\alpha)R_2- \nu\left (\sum_{i=1}^{K_1} p_i - P_s\right ) \\
&\hphantom{\mathcal{L}(\rho,\bm p,\bm q,\alpha,\nu,\mu) = }- \mu\left (\sum_{i=1}^{K_2} q_i -  \sum_{i=1}^{K_1} \rho_i p_i \lambda_{H,i}\right ) \tag{7}
\end{align*}
The Lagrange dual function of (P) can be defined as $g(\alpha, \nu, \mu) = \max \mathcal{L}(\boldsymbol{\rho, p, q},\alpha,\nu,\mu)$ with the dual problem defined as P-Dual $= \underset{\alpha, \nu, \mu \geq 0}\min g(\alpha, \nu, \mu)$. Since the problem is convex, and Slater's condition is satisfied, the duality gap is zero, which implies that both the primal and dual variables can be solved for as a primal-dual pair,$(p_i^\star, q_i^\star, \rho_i^\star,\alpha^\star, \nu^\star, \mu^\star)$.
\subsection{Optimal Solution}
\begin{theorem}
The optimal power allocation, $\bm p$ and $\bm q$ in the first and second hop respectively, and the optimal power splitting ratio, $\boldsymbol{\rho}$, can be obtained in closed-form in terms of the dual variables as
\begin{align*}
&p_i^\star = \Bigg ( \frac{\alpha}{2\nu - 2\mu \rho_i \lambda_{H,i}} - \frac{1}{(1-\rho_i)\lambda_{H,i}} \Bigg )^+ \tag{8a}\\
&q_i^\star = \Bigg ( \frac{1-\alpha}{2\mu} - \frac{1}{\lambda_{G,i}} \Bigg )^+ \tag{8b}\\
&\rho_i^\star = \Bigg [ 1 - \frac{1}{p_i \lambda_{H,i}} \bigg ( \frac{\alpha}{2\mu} - 1\bigg ) \Bigg ]_0^1 \tag{9}
\end{align*}
\end{theorem}
\begin{proof}
Obtained directly through KKT conditions by using~(\ref{Lagrangian}) and setting $\nabla_{p_i} \mathcal{L} = 0$, $\nabla_{q_i} \mathcal{L} = 0$ and $\nabla \mathcal{L}_{\rho_i} = 0$, respectively. Details omitted due to space limitation. Here $(x)^+ = \max(x,0)$, and $[x]_0^1 = \max (\min (x,1),0)$ to ensure implicit constraints $p_i, q_i \geq 0$ and $0 \leq \rho_i \leq 1$ respectively.  
\end{proof}
\begin{remark}
We see that $q_i$ in (8b) is as per conventional water-filling with a constant power level, however, $p_i$ in (8a) has varying power levels for the optimized power allocation. These power levels are varied according to the channel eigenmodes and PS ratios, to achieve a tradeoff between ID and EH. 
\end{remark}
\begin{remark}
In the expression for $\rho_i^\star$ in (9), we see that the term $\frac{1}{\lambda_{H,i} p_i}$ is not constant, therefore the optimal PS ratio is non-uniform. For comparison to the uniform PS scheme in the next section, we evaluate $\nabla \mathcal{L}_{\rho} = 0$ while setting $\rho_i = \rho \ \forall i$. 
\end{remark}
\subsection{Primal Dual Algorithm}
The optimal primal solution in terms of the dual variables in Theorem 1 forms a basis for an efficient primal-dual algorithm to solve (P-eq). Using this optimal primal solution, we obtain the dual function, $g(\alpha, \nu, \mu)$. The problem then reduces to minimizing the dual function in terms of the dual variables. We can use a subgradient based method to solve for the dual minimization problem, where the subgradient terms for the dual variables are as given below. 
\begin{align*}
&\Delta \alpha = \frac{1}{2} \sum_{i=1}^{K_1} \log_2(1+(1-\rho_i)p_i\lambda_{H,i}) - \frac{1}{2} \sum_{i=1}^{K_2} \log_2(1+q_i\lambda_{G,i}) \\
&\Delta \nu = P_s - \sum_{i=1}^{K_1} p_i \ \ \ \ \Delta \mu =  \sum_{i=1}^{K_1} \rho_i p_i \lambda_{H,i} - \sum_{i=1}^{K_2} q_i \tag{10}
\end{align*}
For our implementation, we use the shallow cut ellipsoid method~\cite{Boyd2004}. From the expressions for $p_i^\star$ and $\rho_i^\star$ in (8a) and (9), we see that they are both interdependent. Therefore we initialize the algorithm with  $p_{i,0} \geq 0$, and $0 \leq \rho_{i,0} \leq 1$. As the optimal values for the dual variables are reached using Algorithm 1, the values for the primal variables also converge to their respective optimal values by strong duality.
\begin{algorithm}[t]
\caption{Solution for Rate Optimization Problem}
\textbf{Given:} Channel Matrices, $\boldsymbol{H, G, F}$. Precision, $\epsilon_0$\\
\textbf{Initialize:} Dual variables, $\alpha, \nu, \mu$. Primal variables, $\boldsymbol{p, \rho}$, \\Ellipsoid shape matrix, $P$.\\
\textbf{Begin Algorithm}
\begin{itemize}[leftmargin=*]
\item Use closed form expression in (9) to find $\rho_i^\star$, if $\rho_i^\star$ is $\notin$ (0,1), use boundary conditions, \\
$\circ$ If $\rho_i < 0 \implies \text{set } \rho_i^\star = \epsilon$ \\
$\circ$ Elseif $\rho_i > 1 \implies \text{set } \rho_i^\star = 1 - \epsilon$, where $\epsilon \to 0, \epsilon \ne 0$
\item Waterfill to find $p_i^\star$ from (8a) and $q_i$ from (8b)
\item Check $g(\alpha, \nu, \mu) = \mathcal{L}(\boldsymbol{\rho^\star, p^\star, q^\star},\alpha,\nu,\mu)$\\
$\circ$ If dual variables converge with required precision, \textbf{stop}\\
$\circ$ Else, find subgradients in (10), update dual-variables using ellipsoid method, \textbf{continue}
\end{itemize}
\textbf{End Algorithm}
\end{algorithm}
\section{Numerical Results and Analysis}
In this section, we show the performance of our proposed method for optimal transmission in a MIMO system. We analyze our algorithm's convergence and also present results on the optimality of our non-uniform power splitting scheme. For simulations, path loss exponent $\gamma = 3.2$, noise floor $N_0 = -100$dBm; $d_{sr} = 2\text{m and }d_{sd} = 10\text{m}$ are the distances between S-R and R-D respectively. We show the results for NxNxN MIMO system(s) where $N_s=N_r=N_d=N$. Numerical results are averaged over 5000 independent channel realizations, where each element of the channel matrices, $\mathbf{H}$ and $\mathbf{G}$, is generated as $x_{ij} = (1/d)^{\gamma} u _{ij}$ with $u_{ij}\sim \mathcal{CN} (0,1)$, $d$ being the respective distance.
\begin{figure}[b]
\begin{center}
\textcolor{blue}{\fboxrule=2pt\fbox{\includegraphics[scale=0.4]{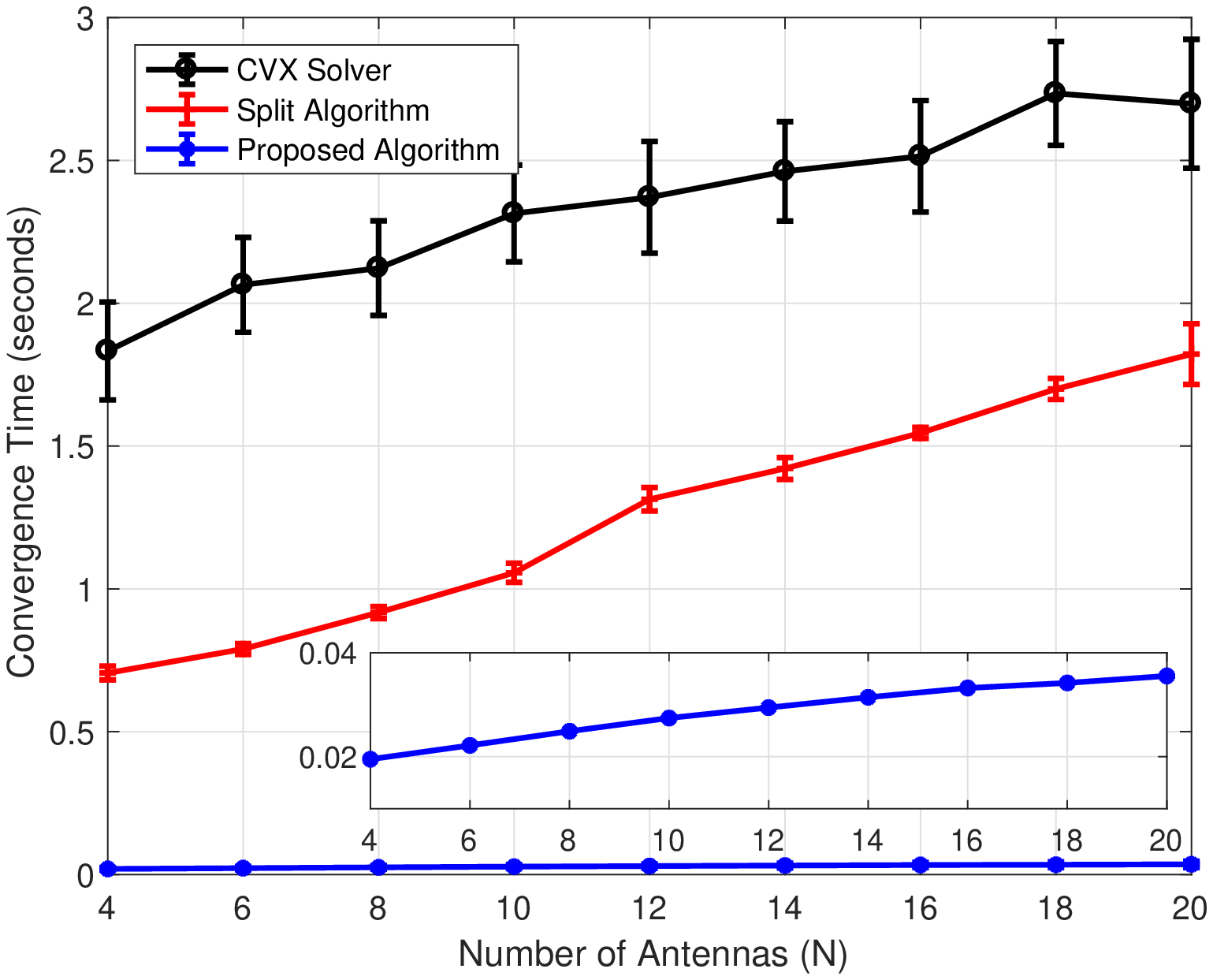}}}
\caption{Proposed primal-dual algorithm converges several orders of magnitude faster than the CVX solver and the split algorithm for NxNxN MIMO system ($P_s = 30$ dBm)}
\label{fig4}
\end{center}
\end{figure}

We compare our primal-dual algorithm, which uses the ellipsoid method for updating the dual variables and closed-form expressions to update the primal variables, with the convex standard solver CVX. The ellipsoid algorithm uses modest storage and computation per step, $O(n^2)$. CVX on the other hand uses the default SDPT3 solver and a heuristic successive approximation method which is several times slower due to its iterative approach~\cite{cvx}, and has a complexity cost of $O(n^3)$~\cite{Alizadeh1998}. Figure~\ref{fig4} shows the stark difference in the convergence time of our proposed algorithm and that of the numerical solver, CVX. We present a zoomed in version for our algorithm to show that it requires only a hundredth fraction of time compared to CVX. For both the solver and the proposed algorithm, the convergence time increases as the number of antennas in the MIMO system are increased. \textcolor{blue}{We also compare the performance to the approach in~\cite{Alouini2016}, where the rate optimization problem for both hops is solved separately as a split problem, and an iterative grid search is used to find the uniform power splitting ratio. The convergence time for the split-algorithm is several times more than our proposed algorithm because of the iterative grid search, however, it is faster than the CVX solver since we implement it using an ellipsoid algorithm as proposed in~\cite{Rui2013}.}

Figure~\ref{fig3} shows a comparison between using optimal non-uniform power splitting and the traditional uniform power splitting for source transmit powers of 25 dBm and 40 dBm. For the higher transmit power of 40 dBm, both schemes deliver comparable performance with gradual increase in throughput as the number of antennas (N) are increased in the NxNxN MIMO system. For the lower transmit power of 25 dBm, the rate gain from non-uniform power splitting is significantly pronounced. This result shows that non-uniform power splitting is beneficial at low source transmit power.
\begin{figure}[t]
\begin{center}
\includegraphics[scale=0.55]{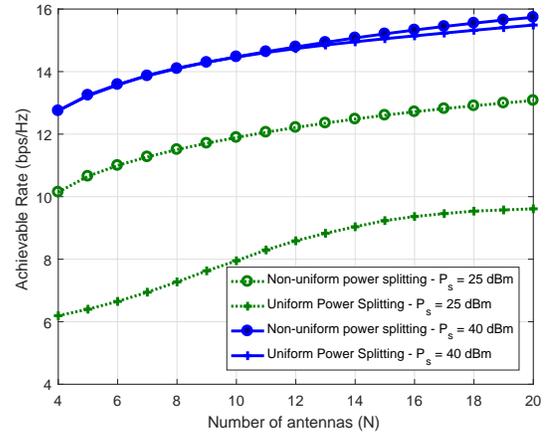}
\caption{Proposed non-uniform power splitting achieves higher rate than uniform power splitting for NxNxN MIMO system at low source transmit power ($N_0 = -100$ dBm)}
\label{fig3}
\end{center}
\end{figure}
\section{Conclusion}
In this letter we formulated a novel optimization problem to solve for the achievable rate of a wirelessly powered MIMO relay system. We characterized the Lagrangian dual problem and designed a primal-dual algorithm for jointly optimizing the source and relay precoders and the relay power splitting ratios. Our analysis showed that the optimal throughput for the MIMO relay channel is achieved using non-uniform power splitting and variable power allocation; and these parameters are interdependent and are affected by the channel eigen-modes. Numerical analysis showed that our proposed algorithm is efficient and runs faster than the CVX solver \textcolor{blue}{and existing iterative algorithms} by several orders of magnitude. Further comparison with standard uniform power splitting revealed that non-uniform power splitting achieves significantly higher rate at low source transmit power, and the achievable rate increases with increased number of antennas in the MIMO system. 
\bibliographystyle{./IEEEtran}
\bibliography{./letterbib}
\end{document}